\documentclass[10pt,final,journal,letterpaper,oneside,twocolumn]{IEEEtran}
\usepackage[T1]{fontenc}
\usepackage[utf8]{inputenc} 
\usepackage{float}
\usepackage{amsmath}
\usepackage{algorithm}       
\usepackage{algorithmic}     
\usepackage{amssymb}         
\usepackage{amsfonts}
\usepackage{mathtools}
\usepackage{amsfonts}
\IEEEoverridecommandlockouts
\usepackage{amssymb}
\usepackage{graphicx}
\usepackage{makeidx}
\usepackage{csquotes}
\usepackage{graphicx}
\usepackage{cite}
\usepackage{xcolor}
\usepackage{subcaption}
\usepackage{graphicx}
\usepackage{balance}
\usepackage{amsthm}
\newtheorem{lemma}{Lemma}
\usepackage{soul}

\begin{document}

\title{Multilayer Non-Terrestrial Networks with Spectrum Access aided by Beyond-Diagonal RIS}
\author{Wali Ullah Khan, \textit{Member, IEEE,} Chandan Kumar Sheemar, \textit{Member, IEEE,} Eva Lagunas, \textit{Senior Member, IEEE,} Xingwang Li, \textit{Senior Member, IEEE,} Symeon Chatzinotas, \textit{Fellow, IEEE,}\\ Petar Popovski, \textit{Fellow, IEEE,} and Zhu Han, \textit{Fellow, IEEE} \thanks{Wali Ullah Khan, Chandan Kumar Sheemar, Eva Lagunas, and Symeon Chatzinotas are with the Interdisciplinary Centre for Security, Reliability, and Trust (SnT), University of Luxembourg, Luxembourg (emails: \{waliullah.khan, chandankumar.sheemar, eva.lagunas, symeon.chatzinotas\}@uni.lu).

Xingwang Li is with School of Physics and Electronics Information
Engineering, Henan Polytechnic University, Jiaozuo 454000, China (email:lixingwangbupt@gmail.com).

Petar Popovski is with the Department of Electronic Systems, Aalborg University, Aalborg, Denmark (email: petarp@es.aau.dk).

Zhu Han is with the Department of Electrical and Computer Engineering at the University of Houston, Houston, TX 77004 USA, and also with the Department of Computer Science and Engineering, Kyung Hee University, Seoul, South Korea, 446-701.  (e-mail: hanzhu22@gmail.com).

The earlier version of this work was accepted for presentation in the 33rd European Signal Processing Conference (EUSIPCO 2025), Palermo Italy, 8-12 September 2025 \cite{khan2025beyond}.
}
}

\markboth{IEEE Transactions (for review)}
{Shell \MakeLowercase{\textit{et al.}}: Bare Demo of IEEEtran.cls for IEEE Journals} 

\maketitle

\begin{abstract}
Sixth-generation (6G) systems are envisioned to fuse terrestrial and non-terrestrial segments to deliver truly global, resilient, and ubiquitous connectivity. Realizing this vision in multilayer non-terrestrial networks (NTNs) is impeded by spectrum scarcity, severe path loss over long links, dynamic channels across orbital/stratospheric layers, and tight energy and hardware constraints on airborne/spaceborne platforms. These challenges motivate propagation-aware solutions that actively shape interference and improve spectral–energy efficiency. Beyond-diagonal reconfigurable intelligent surfaces (BD-RIS) provide advanced control over the wireless propagation environment, which is particularly valuable for multilayer NTNs.  In this work, we study a multi-user NTN in which a satellite serves as the primary network and a high-altitude platform station (HAPS) operates as the secondary network, acting as a cognitive radio. To reduce the cost, complexity, and power consumption of conventional antenna arrays, we equip the HAPS with a transmissive BD-RIS antenna front end. We then formulate a joint optimization problem for the BD-RIS phase response and the HAPS transmit-power allocation under strict per-user interference-temperature constraints. To tackle the resulting highly nonconvex problem, we propose an alternating-optimization framework: the power-allocation subproblem admits a closed-form, water-filling–type solution derived from the Karush–Kuhn–Tucker (KKT) conditions, while the BD-RIS configuration is refined via Riemannian manifold optimization. Simulation results show significant gains in data rate and interference suppression over diagonal RIS–assisted benchmarks, establishing BD-RIS as a promising enabler for future multilayer NTNs.

\end{abstract}



\IEEEpeerreviewmaketitle

\section{Introduction}
 
The integration of non-terrestrial networks (NTNs) into sixth-generation (6G) systems aims to provide uninterrupted global coverage particularly in underserved or remote areas \cite{10716670,11029408,10841371}. Unlike traditional terrestrial networks, where communication largely relies on ground-based base stations (BSs), NTNs leverage satellites, high-altitude platform stations (HAPS), and unmanned aerial vehicles (UAVs) to extend connectivity across diverse geographic environments \cite{sheemar2025joint,iacovelli2024holographic,10938203}. 
However, a fundamental challenge in NTNs is spectrum scarcity \cite{10097680}. Because spectrum resources are limited and long-distance wireless links suffer from significant propagation losses and dynamic channel conditions, efficiently sharing spectrum among multiple NTN layers and terrestrial systems is critical \cite{10938203}. Consider, for example, a disaster recovery scenario where a UAV must rapidly provide connectivity in a region already served by terrestrial or satellite systems. In such cases, the need for efficient spectrum sharing becomes evident. One promising approach to spectrum sharing is cognitive radio (CR) \cite{10791451}. In CR-enabled networks, a secondary network opportunistically reuses the spectrum of a primary network without degrading the service quality of primary users \cite{9514409,10013700,11010845}. This dynamic access mechanism alleviates spectrum scarcity, particularly in dense communication environments, and thus is a natural fit for NTNs.

To further enhance spectrum sharing and mitigate interference, reconfigurable intelligent surfaces (RIS) have emerged as an enabling technology \cite{10557617}. By intelligently controlling the wireless propagation environment, RIS can improve spectrum efficiency, energy efficiency, and communication reliability \cite{9913356}. Recently, beyond-diagonal RIS (BD-RIS) has been proposed as an advanced form of RIS \cite{9514409}. Unlike conventional diagonal RIS (D-RIS), where each element independently adjusts only its phase, BD-RIS employs a reconfigurable impedance network to jointly manipulate the phase-shifts and inter-element coupling. This richer reconfiguration space allows advanced functionalities such as multi-user interference suppression, spatial multiplexing, and improved beamforming gain \cite{10003076,9903905,10584518,9860805,10365519,10138693,10396846}. In terms of spectrum access in NTNs, such capabilities translate into finer-grained control of interference footprints and better spectrum utilization than conventional D-RIS designs.

BD-RIS can also be employed to develop sustainable antenna front-ends, with the motivation being threefold: 1) cost and energy efficiency, since conventional antenna arrays require many RF chains and thus entail high hardware cost and power consumption, whereas BD-RIS operates in a semi-passive manner and substantially reduces both; 2) self-interference mitigation, because reflective BD-RIS can suffer from the interaction of incident and reflected signals on the same side of the surface, while a transmissive layout places the transmitter and receiver on opposite sides and effectively removes this coupling; and 3) enhanced bandwidth and aperture efficiency, as transmissive architectures alleviate the bandwidth limitations and aperture inefficiency typical of reflective designs, yielding broader operational bandwidth and improved effective aperture with corresponding performance gains.

\subsection{Related Works}
The majority of current studies on BD-RIS focus on terrestrial communication scenarios. In the context of single-antenna systems, \cite{10623689} analyzed BD-RIS for wideband channels, deriving capacity expressions, considering channel taps in the frequency domain, and proposing an optimization algorithm that outperformed D-RIS, especially under weak LoS or no static paths. In \cite{10834443}, a BD-RIS-enabled multi-cell SISO framework jointly optimized power allocation and phase shifts via block coordinate descent, achieving higher sum rates than conventional RIS. The work in \cite{li2024beyond} modeled frequency-dependent BD-RIS in OFDM using circuit-based admittance analysis and jointly optimized configuration and power allocation, showing increasing wideband gains with higher circuit complexity.

For multi-antenna systems, \cite{10755162} compared three nearly passive architectures and showed globally passive BD-RIS offered larger rate regions under latency and reliability constraints, albeit with increased complexity. In \cite{10571253}, a DRL-based joint optimization of BS beamforming and BD-RIS matrix improved spectral efficiency over D-RIS, with a trade-off between quantization resolution and complexity. The study \cite{10319662} proposed a closed-form passive beamforming solution and a two-stage BS–BD-RIS design, yielding higher sum rates with reduced computational complexity. In \cite{khisa2024gradient}, a gradient-based meta-learning approach for SIMO-RSMA jointly optimized receive beamforming, BD-RIS matrix, and power, improving sum rate by 22.5\% over D-RIS. Furthermore, \cite{10817282,10817342} utilized BD-RIS in terahertz communication to extend the wireless coverage in indoor and outdoor environments.

In addition to the above studies, \cite{10308579} deployed BD-RIS on the BS side in MIMO and optimized the coefficients and power via a manifold algorithm, improving the minimum spectral efficiency. In \cite{9737373}, a low-complexity BD-RIS allowing inter-element reflection was optimized via AO and SDR, outperforming D-RIS and approaching fully connected designs. The study \cite{10694491} maximized MIMO capacity under unitary and symmetric constraints using a manifold-based iterative algorithm, achieving higher capacity than D-RIS. In \cite{10999443}, a closed-form capacity-maximizing BD-RIS matrix was derived, analyzed under high signal-to-noise ratio (SNR), and validated to outperform prior BD-RIS designs. The work \cite{10200055} used group-connected impedance networks to design scattering matrices that served all users with reduced complexity while maintaining high channel gains. In \cite{zhao2024channel}, geodesic optimization was applied to reshape singular values of MIMO channels, improving rate and power efficiency, especially at low SNR and high dimensions. Finally, \cite{10643263} introduced a multi-sector joint service BD-RIS model, showing up to 100\% rate gains with optimal performance–complexity trade-offs compared to other architectures.

Recently, a few studies have investigated the application of BD-RIS in NTNs. For example, \cite{10716670} employed BD-RIS to enhance communication for satellites by optimizing power allocation and phase shifts to improve spectral efficiency in a single-user scenario. In \cite{khan2024integration}, a transmissive BD-RIS mounted on a UAV was utilized to support multiuser multicarrier communications, with joint optimization of beamforming and power allocation to maximize data rates. Similarly, \cite{khan2025transmissive} considered a transmissive BD-RIS on a satellite platform to assist IoT devices on the ground, achieving sum-rate maximization via coordinated power and phase shift optimization. Overall, the research on BD-RIS applications in NTNs remains nascent and demands further investigation.

\begin{figure*}[!t]
\centering
\includegraphics [width=.7\textwidth]{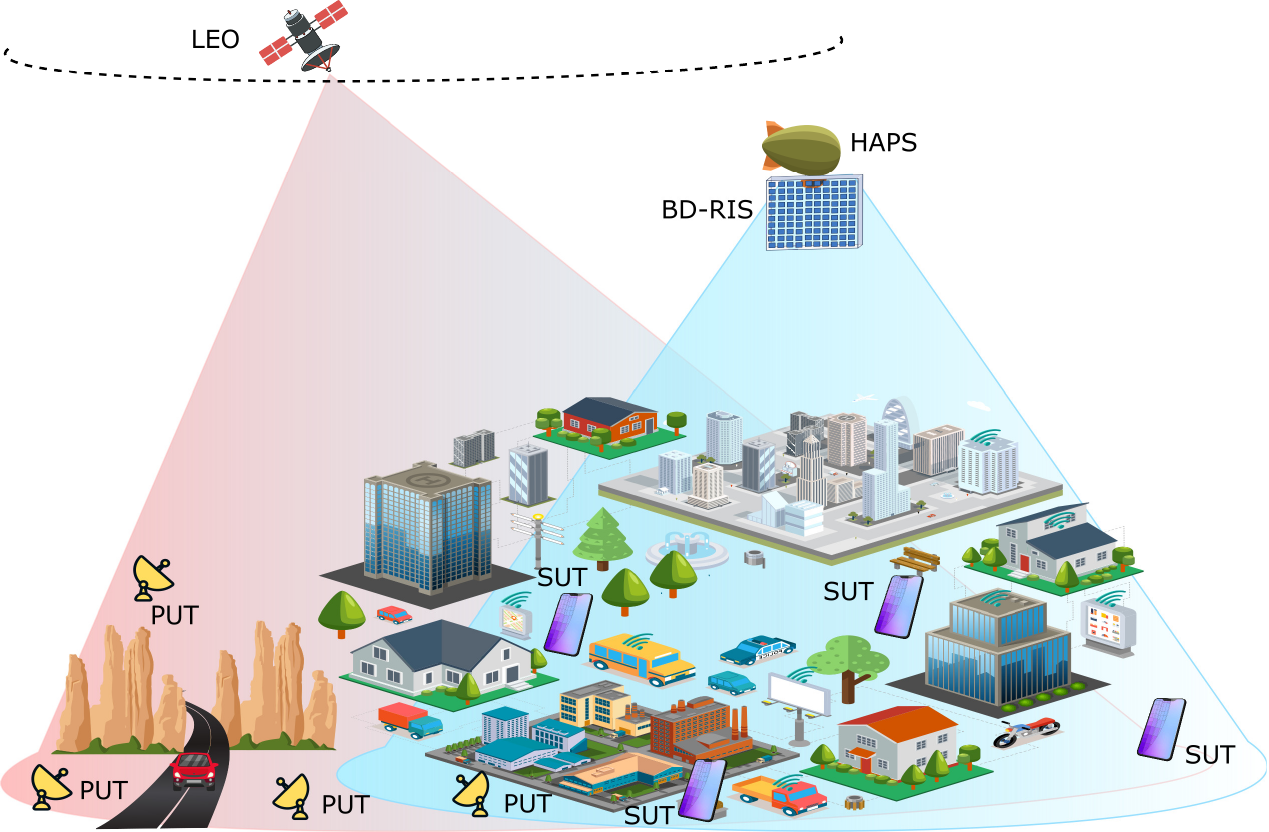}
\caption{Multilayer NTNs empowered by cognitive radio and BD-RIS.}
\label{CRNsm}
\end{figure*}

\subsection{Motivation and Contributions}


NTNs promise to extend reliable broadband connectivity beyond the reach of terrestrial infrastructure, yet spectrum scarcity and cross-tier interference remain central obstacles, especially when multi-layer platforms such as LEO satellites and HAPS must coexist over shared bands. Motivated by this gap, we consider a CR setting in which a secondary HAPS opportunistically reuses the spectrum of a primary LEO system while rigorously protecting the primary service via an interference-temperature constraint. To unlock additional degrees of freedom for interference mitigation and rate enhancement, the HAPS is equipped with a fully connected BD-RIS, whose richer reconfiguration space, enabling joint control of the phase, and inter-element coupling, which offers capabilities unattainable with conventional multi-layer networks.

Firstly, we propose a tractable formulation of spectral-efficiency maximization under practical constraints for the secondary network. The objective couples the HAPS power variables with the BD-RIS matrix multiplicatively, rendering the problem highly non-convex and beyond the reach of standard convex solvers. To address this challenge, we decompose the task via alternating optimization. For fixed BD-RIS, we prove the objective’s concavity in the HAPS transmit powers and derive a closed-form KKT solution with a water-filling-type structure. This solution admits efficient dual updates through monotone bisection, naturally enforcing both the interference caps and the total power constraint. For fixed powers, we optimize the BD-RIS on the Stiefel manifold using a projected Riemannian gradient method with SVD-based retraction to preserve unitarity, alongside a dual ascent step that guarantees feasibility with respect to the primary network’s interference threshold. We further provide a transparent complexity characterization: the power-control stage scales linearly with the number of users, while the BD-RIS stage is dominated by the cubic SVD retraction in the RIS dimension, thereby identifying the true computational bottleneck and guiding parameter choices. Finally, Monte Carlo evaluations demonstrate rapid convergence and consistent spectral-efficiency gains over a same-size conventional D-RIS benchmark across a wide range of HAPS power budgets, interference-temperature thresholds, RIS sizes, and user densities. Collectively, these results position BD-RIS–assisted HAPS as a powerful and practical secondary tier for interference-aware spectrum sharing in multilayer NTNs.

\emph{Paper Organization:} Section \ref{sez_2} provides the multilayer NTN scenario with channel modeling and problem formulation. Section \ref{sez_3} explains a detailed alternating optimization solution, its algorithm, and complexity analysis. Section \ref{sez_4} presents and discusses numerical results and compare them with the conventional framework. Finally, Section \ref{sez_5} concludes this work.

\emph{Mathematical Notations:} In this work, the following mathematical notation is used. Scalars are denoted by lowercase or uppercase letters, while vectors and matrices are represented by bold lowercase and bold uppercase letters, respectively. The transpose, Hermitian transpose, and inverse of a matrix $\mathbf{X}$ are denoted by $\mathbf{X}^\mathrm{T}$, $\mathbf{X}^\mathrm{H}$, and $\mathbf{X}^{-1}$, respectively. Finally, $||\cdot||$ denotes the Frobenius norm.

\section{Multilayer NTNs System Model} \label{sez_2}
A multi-layer NTN-based cognitive radio network is considered, which consists of a primary network and a secondary network, as shown in Fig. \ref{CRNsm}. In the primary network, a LEO satellite communicates with primary user terminals (PUTs) using orthogonal resources. Meanwhile, a secondary network is established, where a HAPS, which is equipped with BD-RIS, transmits signals to multiple secondary user terminals (SUTs) reusing the same spectrum resources as the primary network. In this work, we assume that the signaling required to operate the transmissive BD-RIS is explicit\footnote{Designing the control procedures, protocol aspects, and signaling operations of BD-RIS plays an important role in its performance; however, this is beyond the scope of this work. A detailed discussion of the control plane and practical design guidelines can be found in \cite{10802983}.}. As we have the BD-RIS integrated at the antenna front-end of the HAPS, this control can be enabled by equipping the BD-RIS controller with a dedicated transceiver for signaling control, which can be transmitted over a dedicated wired channel from the baseband up to the analog front-end \cite{10802983}.  We assume that the sets of PUTs and SUTs are denoted $N$ and $K$, where $N=K$. Consequently, BD-RIS consists of $M$ number of reconfigurable elements (REs) that follow a fully connected configuration. Channel state information (CSI) is also assumed to be available both in the primary network and in the secondary network \cite{10097680}. The main objective of this framework is to improve the spectral efficiency (SE) of the secondary network while ensuring that the interference temperature constraint of the primary network is satisfied\footnote{In this work, it is assumed that the primary network has already been optimized prior to the optimization of the secondary network.}. The HAPS is assumed to maintain a quasi-stationary hovering position, as commonly considered in the literature \cite{9515574}. However, Doppler effects may arise due to the motion of surface user terminals (SUTs) or residual platform dynamics, especially in high-frequency bands. Therefore, we adopt a block-fading channel model where the large-scale characteristics of the wireless channel (such as amplitudes and array response) remain constant within each coherence block and change independently across blocks.

\begin{table}[!t]
\centering
\caption{Mathematical symbols and their definition.}
\label{tab:symbols}
\begin{tabular}{|l|l|}
\hline
\textrm{\normalfont Symbol} & \textrm{\normalfont Description} \\
\hline
$M$ & Number of BD-RIS elements \\  \hline
$K$ & Number of SUTs \\  \hline
$N$ & Number of PUTs \\  \hline
$t$ & Time index (fading block index) \\  \hline
$T$ & Total number of fading blocks \\  \hline
$\mathbf{h}_k(t)$ & Channel vector from HAPS to SUT $k$ at time $t$ \\  \hline
$\mathbf{g}_n(t)$ & Channel vector from HAPS to PUT $n$ at time $t$ \\  \hline
$f_k(t)$ & Channel from LEO satellite to SUT $k$ at time $t$ \\  \hline
$\boldsymbol{\Phi}(t)$ & Phase shift matrix  \\  \hline
$p_k(t)$ & Transmit power from HAPS to SUT $k$ at time $t$ \\  \hline
$q_n(t)$ & Transmit power from LEO to PUT $n$ at time $t$ \\  \hline
$x_k(t), x_n(t)$ & Transmitted data symbols from HAPS and LEO \\  \hline
$n_k(t)$ & Additive white Gaussian noise (AWGN) at SUT $k$ \\  \hline
$\sigma^2$ & Noise variance \\ \hline
$I_{\text{th}}$ & Interference temperature threshold \\  \hline
$P_t$ & Maximum total transmit power of HAPS \\ \hline
$Q_p$ & Transmit power of the primary LEO satellite \\ \hline
$R_k(t)$ & Achievable rate for SUT $k$ at time $t$ \\ \hline
$\gamma_k(t)$ & SINR for SUT $k$ at time $t$ \\ \hline
$R_{\text{sum}}(t)$ & Achievable sum rate (ASR) for all SUTs at time $t$ \\ \hline
$\beta$ & Path-loss exponent \\ \hline
$K$ (Rician) & Rician factor (distinct from number of users $K$) \\ \hline
$\hat{h}_k$ & Reference channel gain at $d_0 = 1$ m \\ \hline
$d_k$ & Distance between HAPS and SUT $k$ \\ \hline
$c$ & Speed of light \\ \hline
$f_c$ & Carrier frequency \\ \hline
$f_{D,k}$ & Doppler shift for SUT $k$ \\ \hline
$v_{\text{HAPS}}$ & Relative velocity of HAPS \\ \hline
$\theta_k, \varphi_k$ & Elevation and azimuth angles of arrival for SUT $k$ \\
\hline
\end{tabular}
\end{table}

Let $t \in \{1, 2, \dots, T\}$ denote the time index of the fading block. Consequently, the channel vector from HAPS to SUT $k$ during block $t$ is denoted as $\mathbf{h}_k(t) \in \mathbb{C}^{M \times 1}$ and is modeled as a Rician fading channel, with Doppler-induced phase rotation included within each block. It can be described as\footnote{All mathematical symbols used in this work are listed and described in Table \ref{tab:symbols}}:
\begin{equation}
\mathbf{h}_k(t) = \sqrt{\frac{\hat{h}_k}{d_k^\beta}} \left( \sqrt{\frac{K}{K+1}} \mathbf{h}_k^{\text{LoS}}(t) + \sqrt{\frac{1}{K+1}} \mathbf{h}_k^{\text{NLoS}}(t) \right), \label{eq:h_k}
\end{equation}
where \( K \) in \eqref{eq:h_k} is the Rician factor, \( \hat{h}_k =(c/4\pi f_c d_0)^2\) denotes the reference channel gain at \(d_0=1\)m, \( d_k \) is the distance between the HAPS and SUT \( k \), and \( \beta \) is the path-loss exponent (typically ranging from 2 to 4 depending on the propagation environment).  Moreover, \(\mathbf{h}_k^{\text{LoS}}\) represents the deterministic line of sight component (LoS), and \(\mathbf{h}_k^{\text{NLoS}}\) is the scattering component, which follows a Rayleigh distribution. Following \cite{10839492}, the deterministic LoS component \( \mathbf{h}_k^{\text{LoS}}(t) \) includes a Doppler-induced phase shift which can be expressed as:
\begin{equation}
\mathbf{h}_k^{\text{LoS}}(t) = \big(\mathbf{a}_x(\theta_k, \varphi_k) \otimes \mathbf{a}_y(\theta_k, \varphi_k)\big)\cdot e^{j 2\pi f_{D,k} t T_b} \label{2}
\end{equation}
where \(T_b\) denotes the block duration, \( f_{D,k} \) is the Doppler shift experienced by SUT \( k \), and \( \theta_k(t) \) and \( \varphi_k(t) \) denote the elevation and azimuth angles of arrival, respectively, at time \( t \). Since the users are static, the angles \(\theta_k\) and \(\varphi_k\) are constant; only the Doppler phase rotates over time due to HAPS motion. The Doppler shift experienced by user \( k \) is given by:
\begin{equation}
f_{D,k} = \frac{v_{HAPS} f_c}{c} \cos(\theta_k),
\end{equation}
where \( v_{HAPS} \) is the relative velocity of HAPS, $\theta_k$ is the angle between the HAPS velocity and the LoS to user $k$, assumed constant since the users are static, \( f_c \) is the carrier frequency, and \( c \) is the speed of light. Subsequently, the steering vectors \( \mathbf{a}_x \) and \( \mathbf{a}_y \) along the x-axis and y-axis are given by:
\begin{align}
\mathbf{a}_x(\theta, \varphi) &= \left[1, e^{-j\delta \sin\theta \cos\varphi}, \dots, e^{-j\delta \sin\theta \cos\varphi (M_x - 1)} \right]^T, \\
\mathbf{a}_y(\theta, \varphi) &= \left[1, e^{-j\delta \sin\theta \sin\varphi}, \dots, e^{-j\delta \sin\theta \sin\varphi (M_y - 1)} \right]^T,
\end{align}
where \( \delta = \frac{2\pi f_c q}{c} \), with \( f_c \) being the carrier frequency, \( q \) the spacing between adjacent RIS elements (typically 
\(q\leq\lambda/2\) to avoid grating lobes), and \( c \) the speed of light. Moreover, \( M_x \) and \( M_y \) denote the number of RIS elements along the x-axis and y-axis, respectively, in the BD-RIS system.

The NLoS component \( \mathbf{h}_k^{\text{NLoS}}(t) \sim \mathcal{CN}(\mathbf{0}, \mathbf{I}_M) \) in \eqref{eq:h_k} follows a Rayleigh distribution and is generated independently for each fading block \( t \). Accordingly, the channel from the HAPS to PUT $n$ and from the LEO satellite to SUT $k$ are denoted by $\mathbf{g}_n(t) \in \mathbb{C}^{M \times 1}$ and $f_k(t) \in \mathbb{C}$, respectively. Both $\mathbf{g}_n(t)$ and $f_k(t)$ are modeled as Rician fading channels. While $\mathbf{g}_n(t)$ follows the HAPS–SUT model given in \eqref{eq:h_k} with approximately constant angles due to slow HAPS motion, $f_k(t)$ accounts for the high mobility of the LEO satellite, resulting in time-varying angles and a larger Doppler shift. For simplicity, we do not explicitly redefine the LEO channel model here, as it follows the same mathematical structure as the HAPS–SUT channel. The received signal of SUT $k$ from the HAPS at time slot \( t \) can be expressed as:
\begin{equation}
    y_k(t) = \mathbf{h}_k(t) \mathbf{\Phi}(t)\sqrt{p_k(t)} x_k(t) + f_k(t)\sqrt{q_n(t)}x_n(t) + n_{k}(t),
\end{equation}
where \( x_k(t) \) and \( x_n(t) \) are the i.i.d. unit-variance data streams transmitted from the HAPS and LEO to SUT $k$ and PUT $n$, respectively. In addition, \( \boldsymbol{\Phi}(t) \in \mathbb{C}^{M_x\times M_y} \) denotes the phase shift matrix of the transmissive BD-RIS such that \( \boldsymbol{\Phi}(t)\boldsymbol{\Phi}^H(t) = \mathbf{I}_M \), and \( n_{k}(t) \sim \mathcal{CN}(0, \sigma^2) \) is the additive white Gaussian noise (AWGN). Furthermore, \( p_k(t) \) and \( q_n(t) \) represent the transmit power of HAPS and LEO for SUT $k$ and PUT $n$, respectively. The achievable rate that SUT $k$ receives from HAPS is given as:
\begin{equation}
    R_{k}(t) = \log_2 \left( 1 + \gamma_k(t) \right),
\end{equation}
where \( \gamma_k(t) \) is the signal-to-interference-plus-noise ratio (SINR), defined as:
\begin{equation}
\gamma_k(t) = \frac{|\mathbf{h}_k(t) \mathbf{\Phi}(t)|^2 p_k(t)}{\sigma^2 + |f_k(t)|^2 q_n(t)}.\label{4}
\end{equation}
The achievable sum rate (ASR) for all users at time \( t \) is:
\begin{equation}
    R_{\text{sum}}(t) = \sum\limits_{k=1}^K \log_2 \left( 1 + \gamma_k(t) \right).
\end{equation}
To ensure the QoS of the primary LEO network, the interference power received at the PUTs must not exceed a predefined interference temperature threshold \( I_{\text{th}} \), enforced by a constraint such as:
\begin{equation}
    |\mathbf{g}_n(t) \mathbf{\Phi}(t)|^2 p_k(t) \leq I_{\text{th}},\quad \forall k.
\end{equation}

We aimto maximize the SAR of the secondary network by joint optimization of the power allocation for the HAPS and the phase shift design for the BD-RIS. The optimization problem can be formulated as \( (\mathcal{P}) \):
\begin{align}
 (\mathcal{P}) \quad  & \max_{p_k(t), \mathbf{\Phi}(t)}  \quad R_{\text{sum}}(t) \label{obj}\\
    \text{s.t.} \quad & |\mathbf{g}_n(t) \mathbf{\Phi}(t)|^2 p_k(t) \leq I_{\text{th}},\ \forall k, \label{C1}\\
    & \sum\limits_{k=1}^K p_k(t) \leq P_t, \label{C2} \\
    & \mathbf{\Phi}(t) \mathbf{\Phi}^H(t) = \mathbf{I}_M. \label{C3}
\end{align}
here, \eqref{C1} ensures the QoS of the LEO network by limiting the interference temperature from the HAPS; \eqref{C2} enforces the total HAPS power constraint; and \eqref{C3} imposes the unitary constraint on the BD-RIS phase response. Note that the formulated problem is highly non-convex due to the multiplicative coupling between the transmit power \( p_k(t) \) and the phase shift matrix \( \boldsymbol{\Phi}(t) \). As a result, obtaining a globally optimal solution is extremely challenging. Moreover, the lack of existing studies addressing similar configurations in the literature further motivates the development of novel and scenario-specific algorithmic solutions.

\section{Spectral Efficient Solution}  \label{sez_3}
To address the optimization problem in (\ref{obj}), we adopt an alternating optimization method (AOM). Specifically, the original problem is decomposed into two tractable subproblems: optimizing the transmit power of the HAPS for the $K$ SUTs and configuring the PSD matrix for the BD-RIS. 
For the power allocation task, a closed-form solution with a water-filling-type structure is derived to effectively enhance the system objective. Following that, the PSD of the BD-RIS is optimized using a semidefinite relaxation (SDR) technique to handle the non-convexity of the phase optimization problem. Next, we provide the two-stage optimization.

\subsection{HAPS Power Control}
Given a fixed PSD of BD-RIS at time slot \( t \), problem ($\mathcal{P}$) simplifies to a HAPS power control problem, reformulated as ($\mathcal{P}_1$):
\begin{align} \label{ref_power}
 (\mathcal{P}_1) \quad & \max_{p_{k}(t)} \quad R_{\text{sum}}(t) \\
    \text{s.t.} \quad & |\mathbf{g}_n(t) \mathbf{\Phi}(t)|^2 p_{k}(t) \leq I_{\text{th}}, \quad \forall k,\\
 &  \sum\limits_{k=1}^K p_k(t) \leq P_{t}.
\end{align}
Note that the optimization problem $ (\mathcal{P}_1)$ remains challenging due to the interference constraint and the coupled power variables. To facilitate tractable analysis and solution design, we first establish a key property of the objective function with respect to the transmit power by stating the following lemma.

\begin{lemma}
Given that the PSD of BD-RIS is fixed from the previous iteration at time \( t \), the objective function in ($\mathcal{P}_1$) is concave with respect to the HAPS transmit power for $K$ SUTs.
\end{lemma}

\begin{proof}
To prove this result, we analyze the objective function from ($\mathcal{P}_1$), which can be stated as
\begin{equation}
    f(p_k(t)) = \log_2 \left( 1 + \frac{|\mathbf{h}_k(t) \mathbf{\Phi}(t)|^2 p_k(t)}{\sigma^2 + |f_k(t)|^2 q_n(t)} \right),\quad \forall k.
\end{equation}
To prove concavity, we compute the first and second derivatives of \( f(p_k(t)) \) with respect to the power variable. Applying the chain rule, the first derivative is:
\begin{align}
    \frac{\partial f}{\partial p_k(t)} 
    &= \frac{1}{\ln(2)} \cdot \frac{|\mathbf{h}_k(t) \mathbf{\Phi}(t)|^2}{|\mathbf{h}_k(t) \mathbf{\Phi}(t)|^2 p_k(t) + \sigma^2 + |f_k(t)|^2 q_n(t)}.\label{19}
\end{align}
Then, the second derivative $\frac{\partial^2 f}{\partial p_k(t)^2}$ of \eqref{19} becomes:
\begin{align}
= -\frac{|\mathbf{h}_k(t) \mathbf{\Phi}(t)|^4}{\ln(2) \left( |\mathbf{h}_k(t) \mathbf{\Phi}(t)|^2 p_k(t) + \sigma^2 + |f_k(t)|^2 q_n(t) \right)^2}.
\end{align}
Since all terms are strictly positive, it follows that
\begin{equation}
    \frac{\partial^2 f}{\partial p_k(t)^2} \leq 0, \quad \forall p_k(t) \geq 0,
\end{equation}
which confirms that \( f(p_k(t)) \) is concave with respect to \( p_k(t) \), completing the proof.
\end{proof}
Given the concavity of the objective function, the optimal transmit power allocation for each SUT at time slot \( t \) can be efficiently derived using the KKT conditions. The Lagrangian function of problem $(\mathcal{P}_1)$ can be described as \eqref{eq:Lagrangian_full} on the top of the next page. Next, we compute the partial derivations of \eqref{eq:Lagrangian_full} and set it to zero, as expressed in \eqref{eq:Lagrangian_full2} on the top of the next page.
\begin{figure*}[!t]
\begin{align}
\mathcal{L}\big(\{p_k(t)\},\{\lambda_k(t)\},\nu(t)\big)
 & = \sum\limits_{k=1}^{K} \log_2\!\left( 1 + 
\frac{\big|\mathbf{h}_k(t)\,\mathbf{\Phi}(t)\big|^{2}\,p_k(t)}
{\sigma^{2} + \big|f_k(t)\big|^{2} q_n(t)} \right) 
-\sum\limits_{k=1}^{K} \lambda_k(t)\!\left(\big|\mathbf{g}_n(t)\,\mathbf{\Phi}(t)\big|^{2}\,p_k(t) - I_{\text{th}}\right)\nonumber\\&
- \nu(t)\!\left(\sum_{k=1}^{K} p_k(t) - P_t\right). 
\label{eq:Lagrangian_full}
\end{align}
\end{figure*}
\begin{figure*}[!t]
\begin{align}
\frac{\partial \mathcal{L}\big(\{p_k(t)\},\{\lambda_k(t)\},\nu(t)\big)}{\partial p_k}
 & = \frac{\partial}{\partial p_k}\Bigg[\sum\limits_{k=1}^{K} \log_2\!\left( 1 + 
\frac{\big|\mathbf{h}_k(t)\,\mathbf{\Phi}(t)\big|^{2}\,p_k(t)}
{\sigma^{2} + \big|f_k(t)\big|^{2} q_n(t)} \right) 
-\sum\limits_{k=1}^{K} \lambda_k(t)\!\left(\big|\mathbf{g}_n(t)\,\mathbf{\Phi}(t)\big|^{2}\,p_k(t) - I_{\text{th}}\right)\nonumber\\&
- \nu(t)\!\left(\sum_{k=1}^{K} p_k(t) - P_t\right)\Bigg]=0. 
\label{eq:Lagrangian_full2}
\end{align}
\noindent\rule{\textwidth}{0.4pt}
\end{figure*}
After straightforward calculation, we obtain the following closed-form solution with a water-filling-type structure for the transmit power \( p_k^*(t) \) of SUT \( k \) such as:
\begin{equation}  
    p_k^*(t) = \max \Bigg\{ 0,\frac{1}{\ln(2)\left( \lambda_k(t) c_k(t) + \nu(t) \right)} - \frac{b_k(t)}{a_k(t)} \Bigg\}, \label{p}
\end{equation}
where \( a_k(t) \), \( b_k(t) \), and \( c_k(t) \) in \eqref{p} are defined as:
\begin{equation}
a_k(t) = |\mathbf{h}_k(t) \mathbf{\Phi}(t)|^2.
\end{equation}
\begin{equation}
b_k(t) = \sigma^2 + |f_k(t)|^2 q_n(t).
\end{equation}
\begin{equation}
c_k(t) = |\mathbf{g}_n(t)\mathbf{\Phi}(t)|^2.
\end{equation}

Moreover, the term $1/\ln(2) \left(\lambda_k(t) c_k(t) + \nu(t) \right)$ in \eqref{p} can be interpreted as the ``water level,'' determined by both the interference constraint through the Lagrange multiplier \(\lambda_k(t)\) and the total power constraint through \(\nu(t)\). Accordingly, the term $b_k(t)/a_k(t)$ represents the ``noise floor'' or the inverse effective channel gain of SUT \(k\), which includes the noise variance and the interference from the PUT signal. Power is allocated only to SUT whose channel quality allows the water level to exceed this noise floor, ensuring an efficient use of transmit power while satisfying both interference and power constraints. A mathematical justification for this result is provided in Appendix~A.
For each user \(k\), perform an \emph{inner} bisection on \(\lambda_k\) only if the unconstrained allocation \(p_k(0,\nu)\) exceeds the interference cap \(I_{\text{th}}/c_k\); in that case set bounds \([0,\Lambda_k^{\max}]\) (grow \(\Lambda_k^{\max}\) by doubling until \(p_k(\Lambda_k^{\max},\nu)\le I_{\text{th}}/c_k\)) and solve the monotone scalar equation \(p_k(\lambda_k,\nu)=I_{\text{th}}/c_k\) to tolerance \(\epsilon_\lambda\), which converges in \(\mathcal{O}(\log(1/\epsilon_\lambda))\). Then, in the outer loop, we find $\nu^\star$ by solving 
\begin{equation}
P(\nu) = \sum_{k=1}^K p_k(\lambda_k^\star(\nu),\nu) = P_t,
\label{eq:nu_bisect}
\end{equation} 
if the total power constraint is active or
\begin{equation}
P(\nu) = \sum_{k=1}^K p_k(\lambda_k^\star(\nu),\nu) \leq P_t,
\label{eq:nu_bisect}
\end{equation} 
when the total power constraint is inactive. Note that
$P(\nu)$ is non-increasing in \(\nu\); run an \emph{outer} bisection on \(\nu\) with \(\nu_{\min}=0\), growing \(\nu_{\max}\) by doubling until \(P(\nu_{\max})\le P_t\), then bisect to tolerance \(\epsilon_\nu\) to find \(\nu^\star\) such that \(P(\nu^\star)=P_t\) (or \(\le P_t\) if inactive). Stopping: terminate when all active per-user interference equalities and the sum-power equality are within tolerance, e.g., \(\big|p_k-I_{\text{th}}/c_k\big|\le \epsilon_\lambda\) for active users and \(\big|\sum_k p_k-P_t\big|\le \epsilon_\nu\). When the per-user interference constraint is \emph{known active}, you can directly set the dual variable as:
\begin{equation}
    \lambda_k(t) = \frac{|\mathbf{g}_n(t) \mathbf{\Phi}(t)|^2}{I_{\text{th}}}.
\end{equation}
By substituting this value into \eqref{p}, the optimal transmit power for each SUT \( k \) at time \( t \) can be efficiently computed, ensuring that the interference threshold \( I_{\text{th}} \) is not violated.

\subsection{BD-RIS Optimization}
Subsequently, given the optimal power allocation vector \(\mathbf{p}^*(t) = \{p_1^*(t), p_2^*(t), \dots, p_K^*(t)\}
\), we optimize the BD-RIS phase shift matrix \(\mathbf{\Phi}(t)\). The optimization problem is formulated as

\begin{align}
 (\mathcal{P}_2) \quad  \max_{\mathbf{\Phi}(t)} \quad & \sum_{k=1}^{K} \log_2 \left( 1 + \frac{|\mathbf{h}_k(t) \mathbf{\Phi}(t)|^2 p_k^*(t)}{\sigma^2 + |f_k(t)|^2 q_n(t)} \right) \\
    \text{s.t.} \quad & |\mathbf{g}_n(t) \mathbf{\Phi}(t)|^2 p_k^*(t) \leq I_{\text{th}},\quad \forall k, \\
    & \mathbf{\Phi}(t) \mathbf{\Phi}^H(t) = \mathbf{I}_M.\label{26}
\end{align}
Since the unitary constraint \eqref{26} is non-convex, we employ Riemannian optimization to iteratively update \(\mathbf{\Phi}(t)\) on the Stiefel manifold. To ensure feasibility, we introduce the Lagrange multiplier \(\mu(t)\) to handle interference constraints. It can be stated as follows.
\begin{equation}
    \mathcal{L}(.) = f(\mathbf{\Phi}(t)) - \mu(t) \sum_{k=1}^K \left(|\mathbf{g}_n(t) \mathbf{\Phi}(t)|^2 p_k^*(t) - I_{\text{th}}\right),
\end{equation}
where $\mathcal{L}(.)=\mathcal{L}(\mathbf{\Phi}(t), \mu(t))$ and $f(\mathbf{\Phi}(t))$ is given as:
\[
f(\mathbf{\Phi}(t)) = \sum_{k=1}^{K} \log_2 \left( 1 + \frac{|\mathbf{h}_k(t) \mathbf{\Phi}(t)|^2 p_k^*(t)}{\sigma^2 + |f_k(t)|^2 q_n(t)} \right).
\]
Since \(\mathbf{\Phi}(t)\) lies on the Stiefel manifold, we project the Euclidean gradient onto the tangent space. The Euclidean gradient is given by \eqref{fig:placeholder} on the top of the next page.
\begin{figure*}[!t]
    \centering
\begin{equation}
    \nabla_{\mathbf{\Phi}(t)} \mathcal{L} = \sum_{k=1}^{K} \frac{2 p_k^*(t)}{\ln(2)(\sigma^2 + |f_k(t)|^2 q_n(t) + |\mathbf{h}_k(t) \mathbf{\Phi}(t)|^2 p_k^*(t))} \cdot \mathbf{h}_k^H(t) \mathbf{h}_k(t) \mathbf{\Phi}(t) - 2 \mu(t) \sum_{k=1}^{K} p_k^*(t) \mathbf{g}_n^H(t) \mathbf{g}_n(t) \mathbf{\Phi}(t). \label{fig:placeholder}
\end{equation}
\end{figure*}

The Riemannian gradient is then obtained by projecting the Euclidean gradient onto the tangent space, as shown in \eqref{40} on the next page.
\begin{figure*}[!t]
    \centering
\begin{equation}
    \text{Proj}_{\mathbf{\Phi}(t)}(\nabla_{\mathbf{\Phi}(t)} \mathcal{L}) = \nabla_{\mathbf{\Phi}(t)} \mathcal{L} - \mathbf{\Phi}(t) \left(\mathbf{\Phi}^H(t) \nabla_{\mathbf{\Phi}(t)} \mathcal{L} + \nabla_{\mathbf{\Phi}(t)} \mathcal{L}^H \mathbf{\Phi}(t) \right)/2.\label{40}
\end{equation}
\noindent\rule{\textwidth}{0.4pt}
\end{figure*}
Then, the phase shift matrix is updated as:
\begin{equation}
    \mathbf{\Phi}_{i+1}(t) = \mathbf{\Phi}_i(t) - \eta_i(t) \cdot \text{Proj}_{\mathbf{\Phi}(t)}(\nabla_{\mathbf{\Phi}(t)} \mathcal{L}),
\end{equation}
where \(\eta_i(t)\) is the step size in iteration \(i\). The dual variable \(\mu(t)\) is updated as:
\begin{algorithm}[!t]
\caption{Alternating Optimization for Power Allocation and BD-RIS Phase Shift Design}
\label{alg:AO_RIS_Power}
\begin{algorithmic}[1]
\REQUIRE All channel vectors, noise variance, maximum transmit power of HAPS, transmit power of satellite, interference temperature threshold, number of users, initial phase shift matrix, and maximum iterations.
\ENSURE Optimal power allocation for HAPS $\mathbf{P}^*(t)=[p^*_1,p^*_2,\dots,p^*_K]$, and optimal BD-RIS phase shifts design $\boldsymbol{\Theta}^*$ such that $\boldsymbol{\Theta}(t)\boldsymbol{\Theta}^H(t)=\mathbf{I}_M$.

\STATE \textbf{Initialize}
\STATE Random phase shifts $\boldsymbol{\Theta}^{(0)}$
\STATE Initialize dual variable $\mu^{(0)} \ge 0$
\STATE Set $i = 0$, convergence = \texttt{False}

\WHILE{$i < I_{\text{max}}$ AND convergence = \texttt{False}}
    \STATE \textbf{Step 1: Power allocation for HAPS (Closed-form solution via KKT with water-filling like structure)}
    \STATE Given $\boldsymbol{\Theta}^{(i)}$, solve for $\mathbf{p}^{(i+1)}$ as
    \begin{equation*}
        p_k^{(i+1)} = \left(\frac{1}{\ln(2)\left( \lambda_k(t) c_k(t) + \nu(t) \right)} - \frac{b_k(t)}{a_k(t)}\right)^+ 
    \end{equation*}
    \STATE where $\lambda_k$ ensures $|\mathbf{g}_n(t) \mathbf{\Phi}(t)|^2 p_k(t) \leq I_{\text{th}}$ and $\nu$ guarantees $\sum_k p_k^{(i+1)} \leq P_{t}$.

    \STATE \textbf{Step 2: BD-RIS Phase Shift Design (Stiefel manifold optimization with interference constraints)}
    \STATE Given $\mathbf{p}^{(i+1)}$, form the augmented objective as
    \begin{align*}\small
        \mathcal{L}(.) = f(\mathbf{p}^{(i+1)}, \boldsymbol{\Theta}) - \mu^{(i)} \sum_{k=1}^K \big( |\mathbf{g}_n(t)\boldsymbol{\Theta}|^2 p_k^{(i+1)} - I_{\text{th}} \big)
    \end{align*}
    \STATE where $\mathcal{L}(.)=\mathcal{L}(\boldsymbol{\Theta}, \mu^{(i)})$.
    \STATE Compute the Riemannian gradient of $\mathcal{L}$ on the Stiefel manifold by \eqref{fig:placeholder}
    \STATE Update $\boldsymbol{\Theta}$ via a gradient step and project onto the tangent space by \eqref{40}.
    \STATE Update the phase shift matrix $\boldsymbol{\Theta}$ as
    \begin{equation*}
    \mathbf{\Phi}_{i+1}(t) = \mathbf{\Phi}_i(t) - \eta_i(t) \cdot \text{Proj}_{\mathbf{\Phi}(t)}(\nabla_{\mathbf{\Phi}(t)} \mathcal{L}),
\end{equation*}
    \STATE \textbf{Retract} to the Stiefel manifold via SVD: if $\mathbf{Z} = \mathbf{U}\mathbf{\Sigma}\mathbf{V}^H$ then $\boldsymbol{\Theta}^{(i+1)} \gets \mathbf{U}\mathbf{V}^H$.
    \STATE Update dual variable by \eqref{31}
    \STATE \textbf{Check Convergence}
    \STATE If $\|\mathbf{p}^{(i+1)} - \mathbf{p}^{(i)}\|_2^2 + \|\boldsymbol{\Theta}^{(i+1)} - \boldsymbol{\Theta}^{(i)}\|_F^2 < \epsilon$ and \eqref{con} satisfy, then convergence = \texttt{True}
    \STATE Otherwise, set $i \leftarrow i + 1$
\ENDWHILE

\STATE \textbf{Return} $\mathbf{p}^* = \mathbf{p}^{(i+1)}$, $\boldsymbol{\Theta}^* = \boldsymbol{\Theta}^{(i+1)}$.
\end{algorithmic}
\end{algorithm}
\begin{equation}
\begin{split}
    \mu_{i+1}(t) = \max\Bigg(&0, \mu_i(t) + \alpha \bigg( \\
    & \sum_{k=1}^{K} |\mathbf{g}_n(t) \mathbf{\Phi}_{i+1}(t)|^2 p_k^*(t) - I_{\text{th}} \bigg) \Bigg), \label{31}
\end{split}
\end{equation}
where \(\alpha\) is the learning rate for the dual update. To ensure that the updated \(\mathbf{\Phi}_{i+1}(t)\) remains unitary, we enforce the constraint through the singular value decomposition (SVD) method as follows:
\begin{equation}
    \mathbf{\Phi}_{i+1}(t) = \mathbf{U} \mathbf{V}^H, \quad \text{where } \mathbf{U} \mathbf{\Sigma} \mathbf{V}^H = \text{SVD}(\mathbf{\Phi}_{i+1}(t)).\label{43}
\end{equation}
here in \eqref{43}, $\mathbf{U}$ and $\mathbf{V}$ are unitary matrices, and $\mathbf{\Sigma}$ is a diagonal matrix containing the singular values of $\mathbf{\Phi}_{i+1}(t)$. By reconstructing the matrix as $\mathbf{U} \mathbf{V}^H$, we discard the singular values and project the solution back onto the Stiefel manifold, thereby ensuring that the unitary constraint $\mathbf{\Phi}_{i+1}(t)\mathbf{\Phi}_{i+1}^H(t) = \mathbf{I}_M$ is satisfied. The algorithm continues until convergence is achieved, measured by the thresholds \(\epsilon\) and \(\epsilon_{\text{constraint}}\):
\begin{align}
    \|\mathbf{\Phi}_{i+1}(t) - \mathbf{\Phi}_i(t)\|_F &< \epsilon, \\
    \left|\sum_{k=1}^{K} |\mathbf{g}_n(t)\mathbf{\Phi}_{i+1}(t)|^2 p_k^*(t) - I_{\text{th}} \right| &< \epsilon_{\text{constraint}},\label{con}
\end{align}
where \(\|\cdot\|_F\) denotes the Frobenius norm, which measures the update difference between consecutive BD-RIS matrices \(\mathbf{\Phi}_{i+1}(t)\) and \(\mathbf{\Phi}_i(t)\). The first condition checks for convergence in the optimization variable, while the second ensures the interference constraint is met within a small tolerance. Once both conditions are satisfied, the optimized phase shift matrix \(\mathbf{\Phi}^*(t) = \mathbf{\Phi}_{i+1}(t)\) is passed to the next step of the AO process.

\subsection{Proposed Algorithm and Complexity Analysis}
\subsection*{Proposed Algorithm}
The joint power allocation and BD-RIS phase shift design problem is solved via alternating optimization, where the original problem is first decoupled into two subproblems. Then, a closed-form solution with a water-filling type structure is derived using KKT conditions for HAPS power allocation, given the random BD-RIS phase shift design. Furthermore, the phase shift design for BD-RIS is optimized via a Riemannian approach using the fixed power allocation of HAPS. The Algorithm~\ref{alg:AO_RIS_Power} outlines the complete procedure of our two-step alternating optimization. 

\subsection*{Complexity Analysis}
Given the BD-RIS phase shift matrix from the previous iteration, the power allocation subproblem (P1) is convex and admits a closed-form solution via the KKT conditions, resulting in the water-filling-type expression in (22). The main computational steps involve evaluating the channel gains $a_k(t)$, $b_k(t)$, and $c_k(t)$ for all $K$ users, which requires $\mathcal{O}(KM)$ operations due to the matrix–vector multiplications involving the $M$-element BD-RIS. The dual variables $\{\lambda_k(t)\}$ and $\nu(t)$ are computed using simple bisection searches or closed-form expressions (e.g., (23)), with complexity $\mathcal{O}(K\log(1/\epsilon_b))$, where $\epsilon_b$ is the bisection tolerance. Consequently, the overall complexity of the power allocation step per alternating optimization iteration is $\mathcal{O}(KM + K\log(1/\epsilon_b))$, which scales linearly with both the number of users $K$ and the number of BD-RIS elements $M$.

For the BD-RIS optimization subproblem (P2), the Stiefel manifold-based Riemannian optimization involves several key operations per inner iteration: computing the Euclidean gradient (28), projecting it onto the tangent space (29), performing a gradient descent step (30), retraction via SVD (32), and updating the dual variable $\mu(t)$ (31) for the interference constraints. The Euclidean gradient calculation involves summations over $K$ users, each requiring $\mathcal{O}(M^2)$ operations due to matrix–vector multiplications, leading to a cost of $\mathcal{O}(KM^2)$. The tangent space projection adds $\mathcal{O}(M^2)$ operations, while the SVD-based retraction of an $M\times M$ matrix dominates with $\mathcal{O}(M^3)$ complexity. If $I_{\mathrm{RIEM}}$ denotes the number of Riemannian iterations per alternating optimization step, the total complexity of the BD-RIS optimization is $\mathcal{O}(I_{\mathrm{RIEM}}(KM^2 + M^3))$.

Let $I_{\mathrm{AO}}$ be the number of alternating optimization iterations until convergence. At each alternating optimization iteration, the total cost is the sum of the power allocation complexity and the BD-RIS optimization complexity. Therefore, the overall computational complexity is
\begin{equation*}
    \mathcal{O}\big(I_{\mathrm{AO}} \big[ KM + K\log(1/\epsilon_b) + I_{\mathrm{RIEM}}(KM^2 + M^3) \big] \big).
\end{equation*}
The cubic term $M^3$ from the SVD retraction typically dominates when $M$ is large, making the BD-RIS optimization the main computational bottleneck. However, since $I_{\mathrm{AO}}$ and $I_{\mathrm{RIEM}}$ are moderate in practice due to the fast convergence properties of both subproblems, the proposed alternating optimization framework remains computationally efficient for practical system sizes, especially compared to exhaustive joint optimization over both power and phase variables.
\begin{table}[!t]
\centering
\caption{Simulation Parameters}
\label{tab:sim_params}
\begin{tabular}{|l|l|}
\hline
\textbf{Parameter} & \textbf{Value} \\ \hline
Monte Carlo trials & 1000 \\ \hline
Number of SUTs ($K$) & 4 \\ \hline
Number of secondary users ($K$) & 10 \\ \hline
BD-RIS dimension ($M$) & 64 \\ \hline
HAPS transmit power $(P_t)$ & 35 dBm \\ \hline
Noise power ($\sigma^2$) & $-90$ dBm \\ \hline
LEO transmit power ($Q_p$) & $40$ dBm \\ \hline
Pathloss exponent $(\beta)$ & 2.7 \\ \hline
Rician factor ($K_{\text{Rician}}$) & 10 \\ \hline
Max AO iterations & 40 \\ \hline
Max RIS iterations & 80 \\ \hline
Riemannian step size ($\eta$) & 0.2 \\ \hline
Dual ascent step ($\alpha_\lambda$) & $10^{-1}$ \\ \hline
AO tolerance & $10^{-5}$ \\ \hline
Dual bisection tolerance ($\nu_{\text{tol}}$) & $10^{-6}$ \\ \hline
Max bisection iterations & 60 \\ \hline
Interference threshold ($I_{th}$) & $10^{-2}$ W \\ \hline
\end{tabular}
\end{table}

\begin{figure}[!t]
\centering
\includegraphics [width=.48\textwidth]{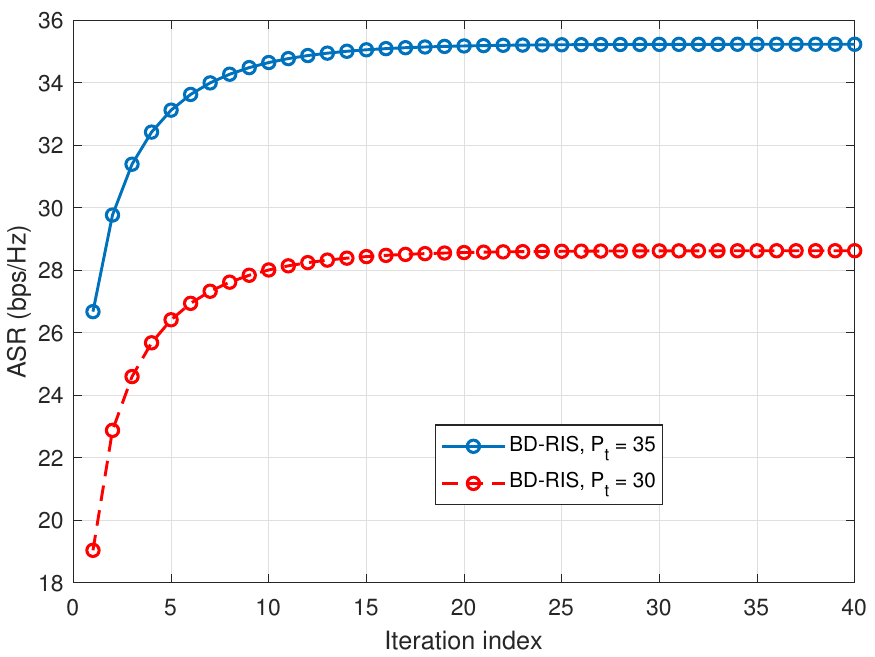}
\caption{Iteration index versus ASR of secondary network, where $I_{th}=0.1$ and $K=4$, and $M=64$ dBm.}
\label{CRNconv2}
\end{figure}

\begin{figure}[!t]
\centering
\includegraphics [width=.48\textwidth]{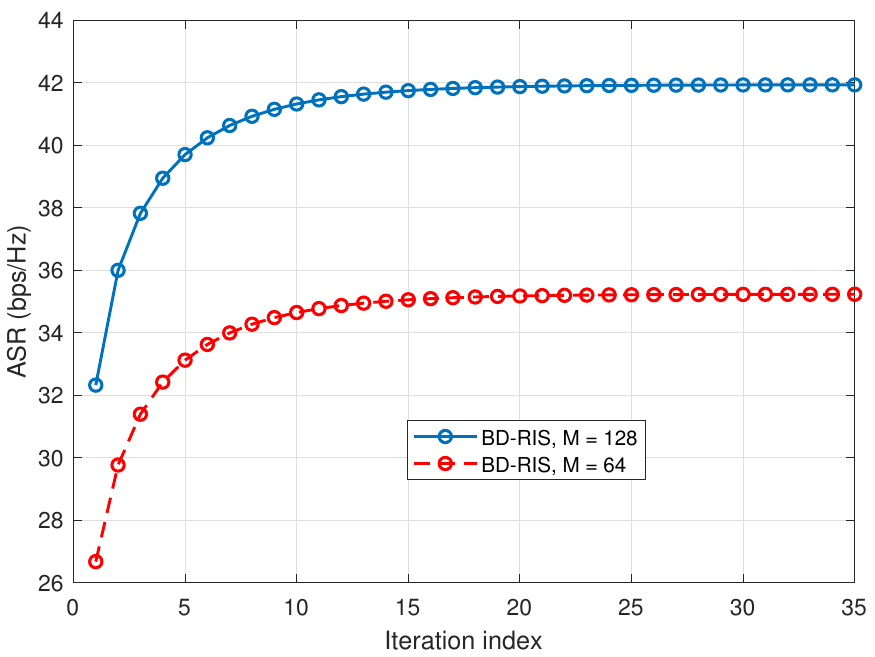}
\caption{Iteration index versus ASR of secondary network, where $I_{th}=0.1$ and $K=4$, and $P_t=35$ dBm.}
\label{CRNconv1}
\end{figure}

\begin{figure}[!t]
\centering
\includegraphics [width=.48\textwidth]{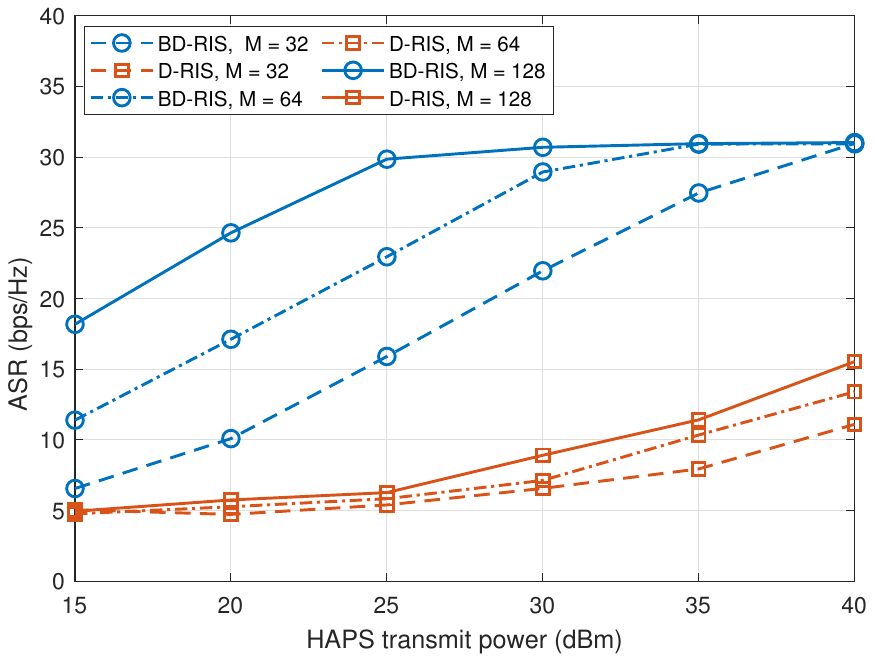}
\caption{HAPS transmit power versus ASR of secondary network for different RIS elements, where $I_{th}=0.01$ and $K=4$.}
\label{CRNratevspower1}
\end{figure}
\section{Numerical Results}  \label{sez_4}
In this section, we evaluate the performance of the proposed cognitive radio-enabled multilayer NTN framework via Monte Carlo simulations. 
Simulation results are presented in terms of spectral efficiency (bits/s/Hz), independent of any specific carrier frequency or bandwidth. Actual throughput for a given system can be directly obtained by multiplying the reported values with the corresponding bandwidth. Each channel realization follows the Rician fading model described in Section~\ref{sez_2}. To demonstrate the performance advantage of the proposed BD-RIS–empowered secondary HAPS network, we compare it with a benchmark system employing a conventional D-RIS of the same size ($M$ elements), under identical system and channel settings. In the simulations, we use normalized Doppler and path-loss parameters without fixing a specific carrier frequency, thereby obtaining spectral efficiency results (bps/Hz). This allows the proposed method to be applied to any carrier frequency and bandwidth, with actual throughput obtained by scaling the spectral efficiency by the chosen bandwidth.

First, we provided the convergence behavior of the proposed alternating optimization algorithm. Fig.~\ref{CRNconv2} and Fig.~\ref{CRNconv1} illustrate the ASR of the secondary network versus the number of iterations for BD-RIS-empowered systems under different parameter settings. In Fig.~\ref{CRNconv2}, the ASR is plotted versus the iteration index for two power budgets of HAPS, i.e., $P_t=30$ and $P_t=35$~dBm, showing that the algorithm converges rapidly within approximately $15$ iterations for low transmission power and 20 iterations for the high transmission power. Similarly, Fig.~\ref{CRNconv1} shows the ASR convergence for different numbers of BD-RIS elements, i.e., $M=64$ and $M=128$, considering the same transmit power. We can see that the large elements of BD-RIS significantly improve the ASR compared to the BD-RIS with a small number of elements due to the higher beamforming gains of the large BD-RIS. In both cases, the results confirm the fast convergence and stability of the proposed algorithm, making it practical for real-time implementations.

The results in Fig.~\ref{CRNratevspower1} illustrate the variation of the secondary network ASR with the HAPS transmit power for different numbers of RIS elements, under the interference threshold $I_{th}=0.01$ and $K=4$ secondary user terminals (SUT). For the proposed BD-RIS--empowered system, the ASR increases significantly with transmit power in the low-to-moderate power region due to the stronger transmission signal gain and the enhanced spatial degrees of freedom provided by larger $M$. However, at higher transmission power levels, the ASR saturates because the interference restriction $I_{th}$ imposed by the primary network becomes the dominant limiting factor, preventing further utilization of the increased power. In contrast, conventional D-RIS achieves substantially lower ASR across all power levels and RIS sizes, with minimal sensitivity to the interference threshold, indicating its inability to fully exploit the available transmit power for secondary transmission. Moreover, for both architectures, increasing the number of reconfigurable elements consistently improves ASR, but the performance gain is much more pronounced for BD-RIS, underscoring its superior capability to leverage additional elements for spectrum-efficient cognitive radio networks.

The results in Fig.~\ref{CRNratevspower2} present the ASR of the secondary network versus the HAPS transmit power for different interference threshold values $I_{th}$, where the number of RIS elements and SUTs are fixed, i.e., $M=64$ and $K=4$, respectively. For the proposed BD-RIS--empowered system, the ASR increases almost linearly in the low-to-moderate power regime when $I_{th}$ is relatively high (e.g., $I_{th}=0.1$), as the secondary transmission can fully exploit the available power. However, for stricter interference constraints (e.g., $I_{th}=0.01$ and $I_{th}=0.001$), the ASR growth rate slows significantly and eventually saturates at high power levels due to the imposed limit on the interference toward the primary network. In contrast, the conventional D-RIS consistently yields much lower ASR across all transmit power levels and shows only marginal sensitivity to $I_{th}$, reflecting its limited ability to adapt to different interference constraints. Overall, the results highlight that BD-RIS can achieve substantially higher spectral efficiency under relaxed interference conditions, while still maintaining an advantage over D-RIS even in highly constrained scenarios.
\begin{figure}[!t]
\centering
\includegraphics [width=.48\textwidth]{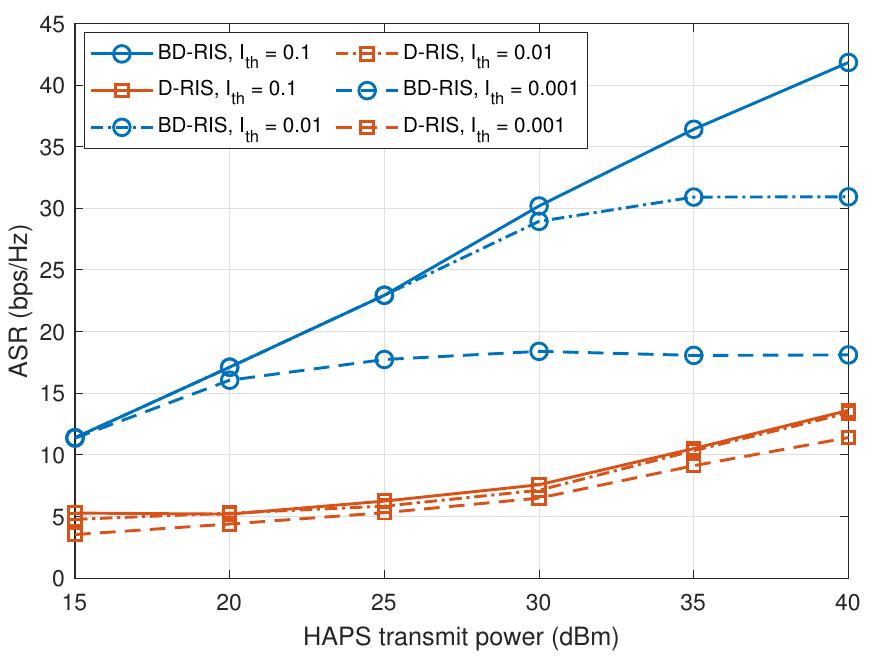}
\caption{HAPS transmit power versus ASR of secondary network for different $I_{th}$ values, where $M=64$ and $K=4$.}
\label{CRNratevspower2}
\end{figure}

\begin{figure}[!t]
\centering
\includegraphics [width=.48\textwidth]{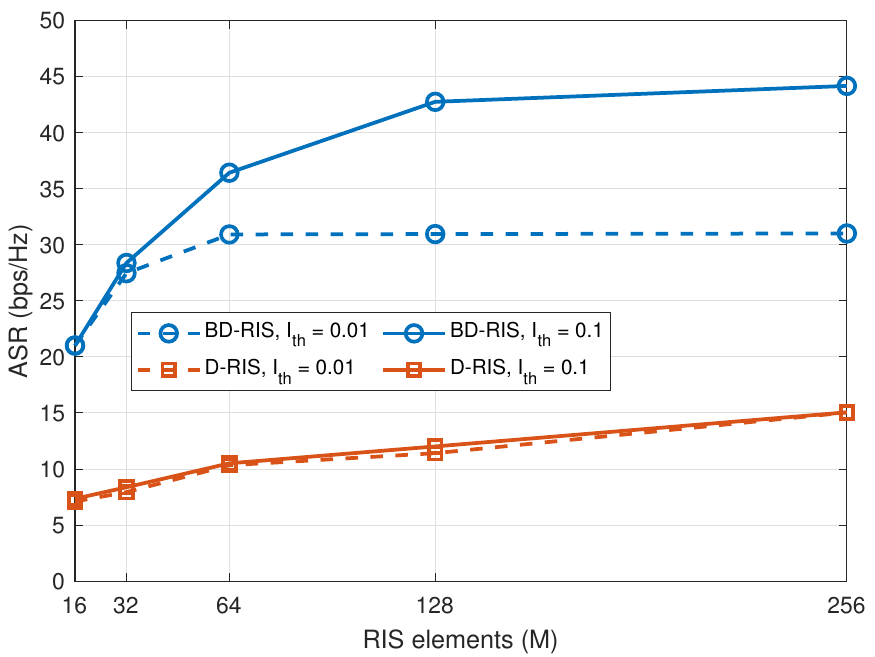}
\caption{RIS elements versus ASR of secondary network, considering $K=4$, and $P_t=35$ dBm.}
\label{CRNratevsRIS3}
\end{figure}


The results in Fig.~\ref{CRNratevsusers4} compare the ASR performance of the secondary network for BD-RIS and conventional D-RIS architectures with respect to the varying number of SUTs, where the HAPS transmit power is set to $P_t = 35$~dBm, the interference threshold is $I_{th} = 0.01$, and the number of reconfigurable elements is $M = 64$. As the number of SUTs increases from $2$ to $8$, both systems exhibit an increasing ASR due to the multiuser diversity gain; however, BD-RIS consistently outperforms D-RIS by a substantial margin across all user counts. This gain is attributed to the additional degrees of freedom provided by the BD-RIS structure, allowing a more efficient transmission design and enhanced interference management in the cognitive radio environment. The advantage becomes particularly pronounced at higher $K$ values, where the richer scattering and advanced beamforming capabilities of BD-RIS more effectively exploit spatial resources, underscoring its scalability for dense multiuser scenarios in multilayer NTNs.

\begin{figure}[!t]
\centering
\includegraphics [width=.48\textwidth]{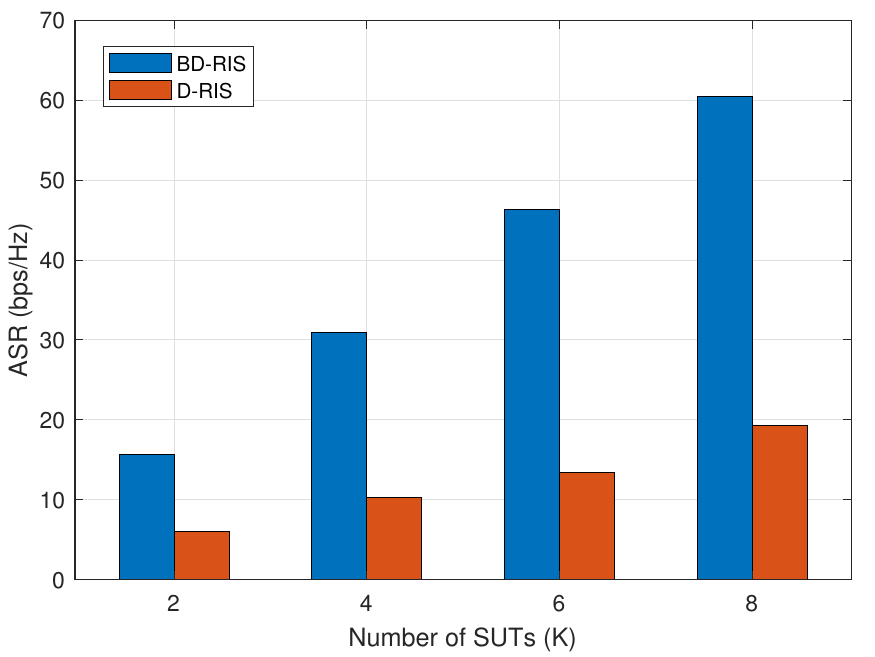}
\caption{SUTs versus ASR of secondary network, considering $I_{th}=0.01$, $M=64$, and $P_t=35$ dBm.}
\label{CRNratevsusers4}
\end{figure}

The results in Fig.~\ref{CRNratevsint5} illustrate the ASR performance of the secondary network for BD-RIS and conventional D-RIS systems against the different values of interference temperature threshold $I_{th}$, where the number of reconfigurable elements is set to $M = 64$ and the HAPS transmit power $P_t = 35$~dBm. Two cases with $K = 4$ and $K = 8$ SUTs for both the BD-RIS and D-RIS systems are considered. As $I_{th}$ increases from $10^{-4}$ to $10^{-1}$, the ASR improves for all scenarios due to relaxation of interference restrictions, which allows for greater allocation of transmit power to the secondary network. Across the entire range of $I_{th}$, BD-RIS consistently outperforms D-RIS, with the performance gap widening at higher $K$ values, where the additional degrees of freedom and enhanced beamforming of BD-RIS enable more efficient interference management. This demonstrates the robustness of BD-RIS in cognitive radio network environments, particularly under looser interference constraints and higher user densities in multilayer NTNs.

\begin{figure}[!t]
\centering
\includegraphics [width=.48\textwidth]{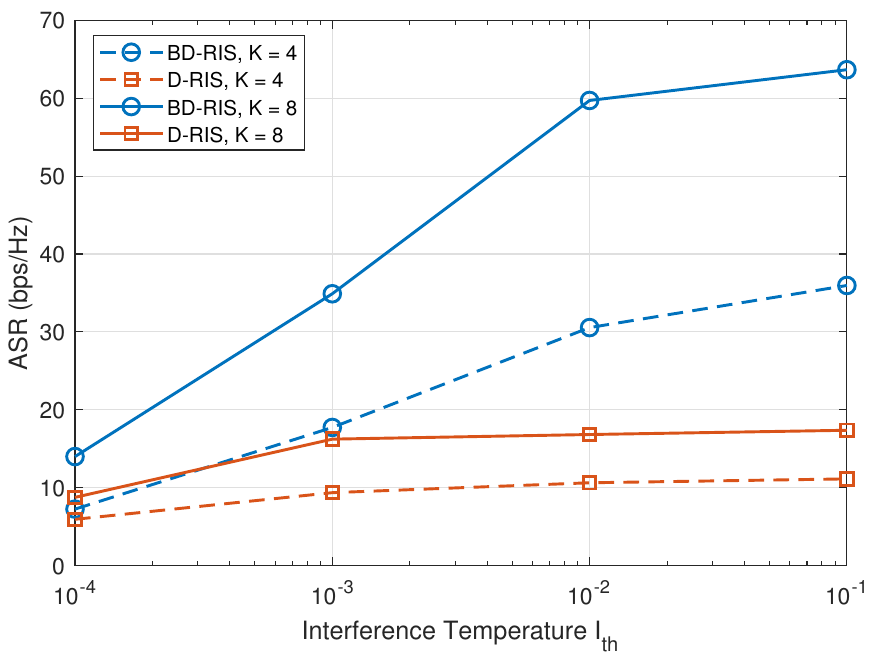}
\caption{Interference temperature threshold versus ASR of secondary network, considering $M=64$, and $P_t=35$ dBm.}
\label{CRNratevsint5}
\end{figure}

Fig.~\ref{CRNratevsusers6} shows the ASR performance of the BD-RIS-empowered system versus the number of SUTs, considering different transmission powers of HAPS, i.e. $P_t = 15$, $25$, and $40$~dBm. As expected, the ASR increases monotonically with the number of SUTs as a result of multi-user diversity gains. In addition, a higher transmit power significantly improves the ASR for all values of $K$, with $P_t = 40$~dBm producing the strongest growth. The performance gap between the curves widens with increasing number of SUTs, indicating that the combined effect of more users and higher transmit power leads to substantial throughput improvements in multilayer NTNs powered by BD-RIS.

\begin{figure}[!t]
\centering
\includegraphics [width=.48\textwidth]{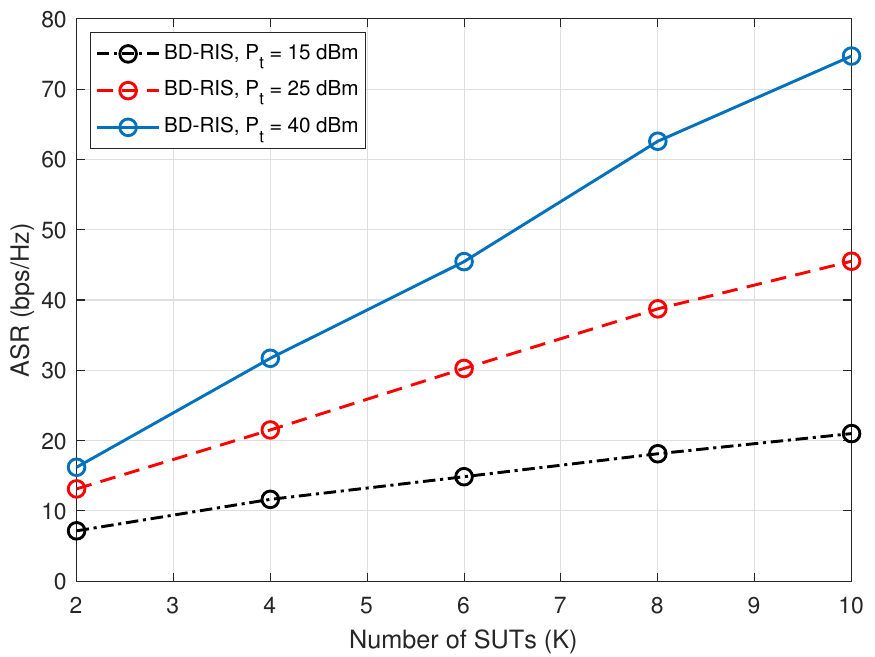}
\caption{SUTs versus ASR of secondary network for different HAPS transmit power, considering $I_{th}=0.1$ and $M=64$.}
\label{CRNratevsusers6}
\end{figure}
\begin{figure}[!t]
\centering
\includegraphics [width=.5\textwidth]{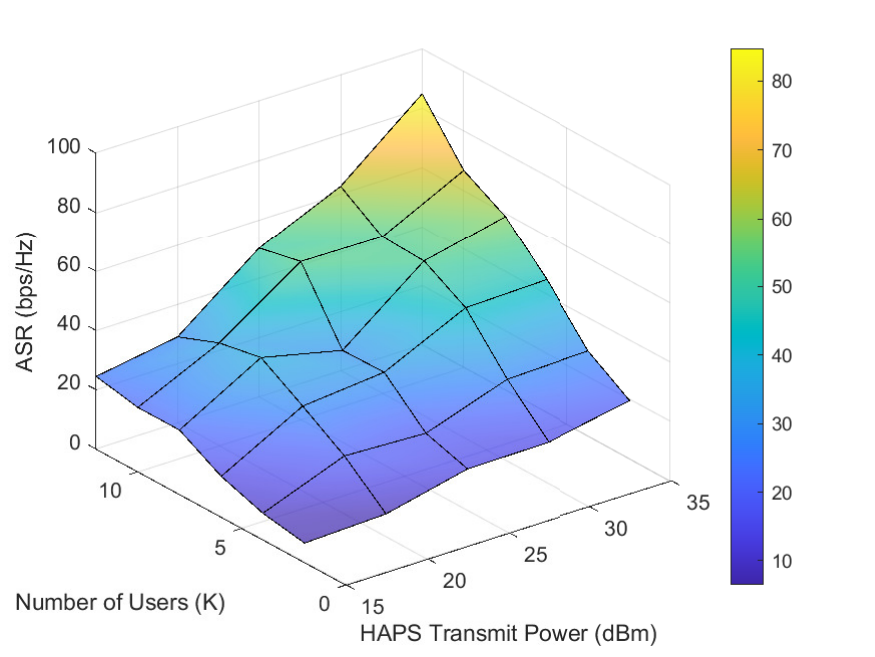}
\caption{ASR of secondary network against varying transmit power and number of SUTs, considering the fixed values of interference temperature and BD-RIS elements, i.e., $I_{th}=0.01$ and $M=64$.}
\label{CRN3d}
\end{figure}
Fig. \ref{CRN3d} illustrates the variation of the ASR of the secondary system with respect to the HAPS transmit power and the number of SUTs. As expected, ASR increases with higher transmit power due to the improved received signal strength and enhanced SINR across the SUTs. Similarly, the ASR grows with the number of SUTs, since more terminals benefit from the BD-RIS-assisted transmission, leading to a higher aggregate system throughput. However, the growth is not strictly linear; at lower power levels, the ASR improvement with additional SUTs is limited due to insufficient transmit power to support all SUTs effectively. At higher power levels, the combined effect of more SUTs and stronger transmission significantly boosts ASR, demonstrating the scalability and effectiveness of the BD-RIS in extensive connectivity in NTNs.
\section{Conclusion}  \label{sez_5}
This paper studied a cognitive radio–enabled multilayer NTN framework where a secondary HAPS is equipped with BD-RIS to enhance spectrum sharing with a primary satellite system. We formulated the joint optimization of HAPS transmit power allocation and BD-RIS phase-shift design to maximize the ASR of the secondary network while satisfying interference-temperature and power budget constraints to maximize the spectral efficiency. To tackle this non-convex problem, we proposed an alternating optimization algorithm that exploits a water-filling–like closed-form solution for power allocation and a Riemannian manifold–based update for the BD-RIS configuration. Extensive Monte Carlo simulations demonstrated that the proposed BD-RIS–empowered HAPS system significantly outperforms a benchmark with a conventional D-RIS of the same size, especially in terms of spectral efficiency and interference management under a stringent cognitive radio environment. The results highlight the potential of BD-RIS technology to unlock additional degrees of freedom in non-terrestrial spectrum sharing and deliver substantial performance gains over conventional D-RIS designs.

\section*{Appendix A: Derivation of Closed-Form Optimal Power Allocation for HAPS}

To prove \eqref{p}, we recall the power allocation problem. Note that we consider the problem of maximizing the ASR under both interference and total power constraints in the presence of \( K \) SUTs. The optimization problem for each time interval \( t \) is formulated as:

\begin{align}
 (\mathcal{P}_1)\quad   \max_{\{p_k(t)\}} \quad & \sum_{k=1}^{K} \log_2 \left( 1 + \frac{|\mathbf{h}_k(t) \mathbf{\Phi}(t)|^2 p_k(t)}{\sigma^2 + |f_k(t)|^2 q_n(t)} \right) \label{eq:opt_power_obj} \\
    \text{s.t.} \quad 
    & |\mathbf{g}_n(t) \mathbf{\Phi}(t)|^2 p_k(t) \leq I_{\text{th}}, \quad \forall k, \\
    & \sum_{k=1}^{K} p_k(t) \leq P_{\max}.
\end{align}
The Lagrangian for the optimization problem of $(\mathcal{P}_1)$ can be given by:
\begin{align}
\mathcal{L}(.) 
&= \sum_{k=1}^{K} \log_2\!\Bigg( 1 + \frac{|\mathbf{h}_k(t)\mathbf{\Phi}(t)|^2\,p_k(t)}{\sigma^2 + |f_k(t)|^2 q_n(t)} \Bigg) \nonumber\\
&\quad - \sum_{k=1}^{K} \lambda_k(t) \Big( |\mathbf{g}_n(t)\mathbf{\Phi}(t)|^2 p_k(t) - I_{\text{th}} \Big) \nonumber\\
&\quad - \nu(t) \Big( \sum_{k=1}^{K} p_k(t) - P_t \Big),
\label{eq:Lagrangian_full_t}
\end{align}
where \(\lambda_k(t)\) and \(\nu(t)\) are the Lagrange multipliers corresponding to the interference and total power constraints, respectively. Now we compute the partial derivative w.r.t. $p_k(t)$. First, we drive the first term (log term) in \eqref{eq:Lagrangian_full_t} as:
\begin{align}
\frac{\partial}{\partial p_k(t)}\log_2 \Big( 1 + \frac{|\mathbf{h}_k(t)\mathbf{\Phi}(t)|^2 p_k(t)}{\sigma^2 + |f_k(t)|^2 q_n(t)} \Big),
\end{align}
where the resultant expression after derivation is given as:
\begin{align}
=\frac{1}{\ln_2}\cdot \frac{|\mathbf{h}_k(t)\mathbf{\Phi}(t)|^2}{\left( \sigma^2 + |f_k(t)|^2 q_n(t) + |\mathbf{h}_k(t)\mathbf{\Phi}(t)|^2 p_k(t) \right)}.
\end{align}
Next, we drive the second term (interference constraint) in \eqref{eq:Lagrangian_full_t} as: 
\begin{align}
\frac{\partial}{\partial p_k(t)} \left[
- \lambda_k(t) ( |\mathbf{g}_n(t)\mathbf{\Phi}(t)|^2 p_k(t))\right].
\end{align}
After derivation, it can be written as:
\begin{align}
=- \lambda_k(t) ( |\mathbf{g}_n(t)\mathbf{\Phi}(t)|^2.
\end{align}
Subsequently, the last term (power constraint) in \eqref{eq:Lagrangian_full_t} can be derived as:
\begin{align}
\frac{\partial}{\partial p_k(t)} \left[
- \nu(t) \sum_{k=1}^{K} p_k(t)\right] = -\nu(t).
\end{align}
Combining all the terms, we can describe the following:
\begin{align}
 \frac{\partial}{\partial p_k(t)} &= \frac{|\mathbf{h}_k(t)\mathbf{\Phi}(t)|^2}{\ln 2 \left( \sigma^2 + |f_k(t)|^2 q_n(t) + |\mathbf{h}_k(t)\mathbf{\Phi}(t)|^2 p_k(t) \right)}\nonumber\\& 
- \lambda_k(t) |\mathbf{g}_n(t)\mathbf{\Phi}(t)|^2 - \nu(t).
\label{57}
\end{align}
Setting the derivative to zero gives the optimality condition as follows:
\begin{align}
&\frac{|\mathbf{h}_k(t)\mathbf{\Phi}(t)|^2}{\ln 2 \left( \sigma^2 + |f_k(t)|^2 q_n(t) + |\mathbf{h}_k(t)\mathbf{\Phi}(t)|^2 p_k(t) \right)} 
\nonumber\\ &= \lambda_k(t) |\mathbf{g}_n(t)\mathbf{\Phi}(t)|^2 + \nu(t).
\end{align}
Now, solving for $p_k(t)$, it can be stated as:
\begin{align}
&\sigma^2 + |f_k(t)|^2 q_n(t) + |\mathbf{h}_k(t)\mathbf{\Phi}(t)|^2 p_k(t) \nonumber\\
&= \frac{|\mathbf{h}_k(t)\mathbf{\Phi}(t)|^2}{\ln 2 \left( \lambda_k(t) |\mathbf{g}_n(t)\mathbf{\Phi}(t)|^2 + \nu(t) \right)}.
\end{align}
After simplification, the value of $p^*_k(t)$ can be computed as:
\begin{align}
p_k(t) 
&= \frac{1}{\ln 2 \left( \lambda_k(t) |\mathbf{g}_n(t)\mathbf{\Phi}(t)|^2 + \nu(t) \right)} \nonumber\\&
- \frac{\sigma^2 + |f_k(t)|^2 q_n(t)}{|\mathbf{h}_k(t)\mathbf{\Phi}(t)|^2}.\label{60}
\end{align}
Finally, enforcing non-negativity, \eqref{60} can be re-expressed as:
\begin{equation}  
    p_k^*(t) = \max \Bigg\{ 0,\frac{1}{\ln(2)\left( \lambda_k(t) c_k(t) + \nu(t) \right)} - \frac{b_k(t)}{a_k(t)} \Bigg\}, \label{61}
\end{equation}
where \( a_k(t) \), \( b_k(t) \), and \( c_k(t) \) in \eqref{61} are defined as:
\begin{equation}
a_k(t) = |\mathbf{h}_k(t) \mathbf{\Phi}(t)|^2.
\end{equation}
\begin{equation}
b_k(t) = \sigma^2 + |f_k(t)|^2 q_n(t).
\end{equation}
\begin{equation}
c_k(t) = |\mathbf{g}_n(t)\mathbf{\Phi}(t)|^2.
\end{equation}
Thus, \eqref{p} is proved.

\bibliographystyle{IEEEtran}
\bibliography{Wali_EE}

\end{document}